\documentclass{article}
  \usepackage[utf8]{inputenc}
  \usepackage[hyperref,standard,thmmarks]{ntheorem}
	\renewtheorem{theorem}{Theorem}[section]
	\renewtheorem{proposition}[theorem]{Proposition}
	\renewtheorem{lemma}[theorem]{Lemma}
	\renewtheorem{corollary}[theorem]{Corollary}
	\theorembodyfont{\upshape}
	\renewtheorem{example}[theorem]{Example}
	\renewtheorem{remark}[theorem]{Remark}
	\renewtheorem{definition}[theorem]{Definition}
  \usepackage[pdftex]{hyperref}
  \usepackage{
    tikz,
    xspace,
    amsmath,
    amssymb,
    cleveref,
  }

  \newcommand{\NN}{\mathbb N}

  \newcommand{\B}{\mathcal B}

  \newcommand{\lb}{\mathrm{lb}}
  
  \newcommand{\XX}{\mathbf{X}}
  \newcommand{\YY}{\mathbf{Y}}

  \newcommand{\Time}{\mathrm{time}}
  \newcommand{\length}[1]{\left|#1\right|}
  \newcommand{\flength}[1]{\left|#1\right|}
  \newcommand{\sdzero}{\textup{\texttt{0}}}
  \newcommand{\sdone}{\textup{\texttt{1}}}
  \newcommand{\albe}{\Sigma}

  \newcommand{\bigo}{\mathcal O}
  \newcommand{\demph}[1]{\textbf{#1}}

  \newcommand{\dom}{\mathrm{dom}}
  \newcommand{\str}{\mathbf}
  \newcommand{\reg}{\albe^{**}}

  \newcommand{\id}{\operatorname{id}}

  \newcommand{\FLR}{\operatorname{FLR}}
  
  \newcommand{\timef}{\operatorname{time}}
  \newcommand{\opt}{\operatorname{OPT}}
  \newcommand{\PSC}{\operatorname{PSC}}
  \renewcommand{\omega}{\NN}

\title{Polynomial running times\\for polynomial-time oracle machines}
\author{Akitoshi Kawamura and Florian Steinberg}
\date{}
\begin{document}
  \maketitle
  \begin{abstract}
    This paper introduces a more restrictive notion of feasibility of functionals on Baire space than the established one from second-order complexity theory.
    Thereby making it possible to consider functions on the natural numbers as running times of oracle Turing machines and avoiding second-order polynomials, which are notoriously difficult to handle.
    Furthermore, all machines that witness this stronger kind of feasibility can be clocked and the different traditions of treating partial functionals from computable analysis and second-order complexity theory are equated in a precise sense.
    The new notion is named \lq strong polynomial-time computability\rq, and proven to be a strictly stronger requirement than polynomial-time computability.
    It is proven that within the framework for complexity of operators from analysis introduced by Kawamura and Cook the classes of strongly polynomial-time computable functionals and polynomial-time computable functionals coincide.
  \end{abstract}
  \tableofcontents
\section{Introduction}
  Modern applications of second-order complexity theory the field of computable analysis almost exclusively use time-restricted oracle Turing machines to define and argue about the class of polynomial-time computable functionals \cite[etc.]{MR2275414,kawamuraphd,higherorder,MR3239272,postive,CIE2016,ArXiV}.
  The acceptance of this model of computation goes back to a result by Kapron and Cook \cite{MR1374053} that characterizes the class of basic feasible functionals introduced by Mehlhorn \cite{MR0411947}.

  There are several reasons for the popularity of this model of computation.
  Firstly, it intuitively reflects what a programmer would require of an efficient program if oracle Turing machines are interpreted as programs with subroutine calls.
  I.e. the time taken to evaluate the subroutine is not counted towards the time consumption (the oracle query takes one time step) and if the result is complicated the machine is given more time for further operations.
  Secondly, it is superficially quite close to classical polynomial-time computability:
  There is a type of functions that take sizes of the inputs and return an allowed number of steps.
  A subclass of these functions are considered polynomial, or \lq fast\rq\ running times.

  On closer inspection, however, the second-order framework introduces a whole bunch of new difficulties:
  Running times of oracle Turing machines, and also the functions that are considered polynomial running times, are functions of type $\omega^\omega\times \omega\to \omega$.
  These so-called second-order polynomials are a lot less well-behaved than their first-order counterparts.
  There are no normal-form theorems, structural induction turns out to be complicated, there is no established notion of degree and so on \cite{MR3219039,Kawamura:2016:CTC:2933575.2935311}.
  Even worse:
  Second-order polynomials turn out to not be time-constructible \cite{ArXiV}.

  The framework introduced by Kawamura and Cook \cite{Kawamura:2012:CTO:2189778.2189780} addresses this problem by restricting to length-monotone string functions, thereby forcing time-constructibility of second-order polynomials.
  However, it has been argued that the restriction to length-monotone string functions seems to be an unnatural one in practice \cite{ArXiV} and that it is too restrictive to reflect some situations from practice \cite{Aminimal}.
  Thus, this paper investigates different solutions to the same set of problems.
  \subsection*{The content of this paper}
    The first part of the paper investigates the boundaries of the polynomial-time framework.
    Descriptions of second-order polynomials are introduced as a replacement of a normal-form theorem which currently seems to be out of reach.
    In particular they can be used to obtain polynomial majorants:
    For any second-order polynomial there is a polynomial and a number such that the values of the second-order polynomial can be bounded by an easy formula only involving these.
    The polynomial majorants come in handy the later parts of the paper.
    
    Then complexity of partial functionals is investigated.
    The traditions of how to handle partiality differ a lot between computable analysis and second-order complexity theory.
    While the former tends to avoid assumptions about functionals outside of their domain and would in particular not make restrictions on the number of steps a machine may take on elements outside of the domain of the functional it computes, the latter usually requires the existence of total polynomial-time computable extensions.
    The two corresponding classes are introduced and proven to be actually distinct.
    That these classes can be separated can be considered to be a very strong version of the statement that second-order polynomials are not time-constructible.
    Finally, it is proven that in the most important example of use of an intermediate of the two conventions, namely the framework for complexity of operators in analysis as introduced by Kawamura and Cook, could have equivalently used the convention from second-order complexity theory.

    The second part of the paper presents a restriction on the behavior of oracle machines such that use of running times of higher type is not necessary anymore.
    It proves that the corresponding class of functionals, which are named \lq strongly polynomial-time computable functionals\rq, is a subclass of the class of polynomial-time computable functionals.
    It provides an example of a functional that is polynomial-time computable but not strongly polynomial-time computable.
    The example is not a natural example, but there are candidates for a more natural examples.

    Finally the paper presents some evidence that strong polynomial-time computability is more compatible with partial functionals and proves that within the framework for complexity of operators in analysis introduced by Kawamura and Cook, it is equivalent to polynomial-time computability.
    In particular Kawamura and Cook could have fully committed to the traditions of computable analysis in the definitions of their framework for complexity of operators from analysis.
    A functional whose domain is contained in the length-monotone functions is polynomial-time computable if and only if it is strongly polynomial-time computable.

  \subsection*{Conventions}
    Fix the finite alphabet $\Sigma:= \{\sdzero,\sdone\}$ and let $\Sigma^*$ denote the set of finite binary strings.
    Elements of $\Sigma^*$ are denoted by $\str a$, $\str b$, \ldots.
    The set of non-negative integers is denoted by $\NN$, natural numbers are denoted by $n$, $m$, \ldots and sometimes other letters.
    We identify the Baire space with the set $\B:= (\Sigma^*) ^{\Sigma^*}$ of string functions.
    Elements of $\B$ are denoted as $\varphi$, $\psi$, \ldots.
    We assume the reader to be familiar with the notions of computability and complexity theory for elements of the Baire space introduced via Turing machines.

    We call functions of type $\B\to\B$, i.e. functions from Baire space to Baire space, functionals.
    To compute functionals, this paper uses oracle Turing machines:
    An oracle Turing machine $M^?$ is a Turing machine that has an additional oracle query tape and an oracle query state.
    For an arbitrary $\varphi\in\B$, we obtain a string function $M^\varphi$ as follows:
    If the computation of $M^?$ (with oracle $\varphi$ and) on input $\str a$ enters the query state,
    the content of the oracle query tape, say $\str b$, is replaced with the value $\varphi(\str b)$.
    Afterwards the computation continues and if it terminates, its return value is used as the value $M^\varphi(\str a)$ of the string function $M^\varphi$ on $\str a$.
    Note that the string function $M^\varphi$ may not be defined everywhere, as the run of the machine may diverge.
    We say that an oracle machine $M^?$ computes a partial functional $F\colon{\subseteq}\B\to\B$ if $M^\varphi = F(\varphi)$ holds for all elements of the domain of $F$.

    For measuring the time it takes a machine $M^?$ to compute its value $M^\varphi(\str a)$ on oracle $\varphi$ and input $\str a$, overwriting the oracle query $\str b$ with $\varphi(\str b)$ is considered to be done in one time step does not move the reading/writing head.
    The time it takes the machine $M^?$ to terminate with oracle $\varphi$ and on input $\str a$ is denoted by $\Time_{M^\varphi}(\str a)\in\NN$.

\section{Second-order complexity and relativization}
  Functionals are objects of type $\B\to\B$.
  For complexity considerations it is more natural to consider oracle Turing machines to compute objects of type $\B\times\Sigma^*\to\Sigma^*$.
  Each functional can be regarded an object of this type via currying: Instead of a functional $F\colon\B\to\B$ consider the mapping $\tilde F:\B\times\Sigma^*\to\Sigma^*$ defined by $\tilde F(\varphi,\str a):= F(\varphi)(\str a)$.
  In this setting, both $\varphi$ and $\str a$ should be considered inputs, and the time an oracle machine is granted should increase with the \lq size\rq\ of both inputs.
  It is clear what the size of the string input is, and the next definition fixes a notion of size for the oracles.
  \begin{definition}
    Let $\varphi\in\B$ be a string function.
    Its \demph{size function} $\flength\varphi\colon\omega\to\omega$ is defined by
    \[ \flength{\varphi}(n):= \max\{\length{\varphi(\str a)}\mid \length{\str a}\leq n\}. \]
  \end{definition}
  Running times take a size of a string function and a size of a string and return an allowed number of steps, therefore they are objects of the type $\omega^\omega\times \omega\to \omega$.
  However, not all such functions should be eligible as running times.
  For instance:
  As the inputs get bigger, the time granted to the machine should not decrease, at least as long as the functional size argument is monotone and therefore actually turns up as size of a string function.
  \begin{definition}
    We call a function $T\colon\NN^\NN\times\NN\to\NN$ a \demph{running time} if whenever $l$ and $l'$ are monotone and $l$ is point-wise bigger than $l'$, then also $T(l,\cdot)$ and $T(l',\cdot)$ are monotone and the latter is point-wise bigger than the former.
  \end{definition}
  A running time $T$ is \demph{a running time for an oracle machine $M^?$} if for any oracle $\varphi$ and string $\str a$, the run of $M^\varphi$ on input $\str a$ terminates within $T(\flength\varphi,\length{\str a})$ steps.
  That is if
  \begin{equation}\tag{RT}\label{eq:RT}
    \forall\varphi\in\B,\forall\str a\in\Sigma^*\colon\Time_{M^\varphi}(\str a)\leq T(\length{\varphi},\length{\str a}).
  \end{equation}
  
  It is not a priori clear what running times should be considered polynomial.
  The class of second-order polynomials is the smallest class of functions $P\colon\omega^\omega\times \omega\to\omega$ such that:
  \begin{itemize}
    \item All of the functions $(l,n)\mapsto p(n)$ are contained, where $p$ is a polynomial with natural numbers as coefficients.
  \end{itemize}
  And which is closed under the following operations:
  \begin{itemize}
    \item Whenever $P$ and $Q$ are contained, then so is their point-wise sum $P+Q$.
    \item Whenever $P$ and $Q$ are contained, then so is their point-wise product $P\cdot Q$.
    \item Whenever $P$ is contained then so is the function $P^+$ defined by 
    \[ P^+(l,n) := l(P(l,n)). \]
  \end{itemize}
  It is easily checked that any second-order polynomial fulfills the requirement we imposed on running times.
  \begin{definition}
    A functional on Baire space is called \demph{polynomial-time computable} if it is computed by an oracle Turing machine $M^?$ that has a second-order polynomial $P$ as running time.
  \end{definition}
  The above definition is based on a characterization by Kapron and Cook of the class of basic feasible functionals originally introduced by Mehlhorn.

  It is not obvious from the definition that the class of polynomial-time computable functionals is closed under composition.
  To see that this still holds true, we need the following two closure properties of the set of second-order polynomials:
  \begin{lemma}
    Whenever $P$ and $Q$ are second-order polynomials, then so are
    \[ (l,n)\mapsto P(Q(l,\cdot),n) \quad\text{and}\quad (l,n)\mapsto P(l,Q(l,n)). \]
  \end{lemma}
  The proof can be done via a tedious but straight-forward induction on the term structure of second-order polynomials.
  Another, more elegant proof is provided in \Cref{resu:compositions}.

  \begin{proposition}\label{resu:composition}
    Let $F\colon\B\to\B$ and $G\colon\B\to\B$ be functionals that can be computed within times $T$ resp.\ $S$.
    Then $F\circ G$ can be computed in time
    \[ (l,n) \mapsto C\cdot(T(S(l,\cdot),n) + S(l,T(S(l,\cdot),n))\cdot T(S(l,\cdot),n)) \]
    for some $C\in\NN$.
    In particular, the polynomial-time computable functionals are closed under composition.
  \end{proposition}
  \begin{proof}
    Let $M^?$ and $N^?$ be machines that compute the operators $F$ and $G$ and run in times $T$ and $S$.
    Consider the oracle machine $_MN^?$ that proceeds as follows:
    On oracle $\varphi$ and input $\str a$ it follows the computation of $M^?$ on input $\str a$ but with the commands for oracle query tape replaced by the commands for an unused memory tape.
    Each time $M^?$ poses an oracle query, instead of entering the query state it starts to carry out the steps that $N^?$ would do on input $\str b$, where $\str b$ is the content of the memory tape the oracle tape was replaced with.
    Once the machine $N^?$ terminates it switches back to following the steps of $M^?$.
    Once $M^?$ terminates also $_MN^?$ terminates.
    This machine obviously computes $F\circ G$.

    To see that the machine finishes within the specified time, note that, since $N^?$ computes $G$, the steps of the machine $_MN^?$ that copy the behavior of $M^?$ are identical with the steps $M^?$ carries out with oracle $G(\varphi)$ and input $\str a$.
    Since $S$ is a running time of $N^?$ (and in particular a running time), it holds that
    \[ \flength{G(\varphi)}(n) = \max\{\length{M^\varphi(\str a)}\mid \length{\str a}\leq n\}\leq \max\{\Time_{M^\varphi}(\str a)\mid \length{\str a}\leq n\} \leq S(\flength\varphi,n). \]
    Furthermore, $T$ is a running time of $M^?$ and therefore at most $T(S(\flength\varphi,\cdot),\length{\str a})$ steps are spent carrying out the operations of the machine $M^?$.
    
    In particular, all oracle queries $M^?$ can have at most $T(S(\flength\varphi,\cdot),\length{\str a})$ bits.
    Due to $S$ being a running time of $N^?$ the number of steps that are carried out simulating $N^?$ each time $M^?$ asks an oracle query $\str b$ with $\length{\str b}\leq T(S(\flength\varphi,\cdot),\length{\str a})$ is bounded by
    \[ \Time_{N^\varphi}(\str b) \leq S(\length\varphi,\length{\str b})\leq S(\flength\varphi,T(S(\flength\varphi,\cdot),\length{\str a})). \]
    Furthermore, $M^?$ can ask at most $T(S(\flength\varphi,\cdot),\length{\str a})$ oracle queries and thus the total number of steps that are spent simulating $N^?$ is bounded by
    \[ T(S(\flength\varphi,\cdot),\length{\str a})\cdot S(\flength\varphi,T(S(\flength\varphi,\cdot),\length{\str a})) \]
    Adding the number of steps that are carried out when simulating $M^?$ and $N^?$ respectively, and accounting for the additional steps to return heads to the beginning of tapes etc., leads to the time bound from the statement.
  \end{proof}
  Another property of polynomial time computable functionals that should be mentioned is that they preserve the class of polynomial-time computable functions.
  This can easily be checked by combining the program of a polynomial-time machine computing the function with the program of a polynomial-time oracle Turing machine computing the functional.
  
  Second-order polynomials were introduced as functions of type $\omega^\omega\times\omega\to\omega$.
  This is natural since they are considered running times.
  However, it also regularly leads to difficulties:
  It is not clear how to decide equality of two second-order polynomials from the construction procedures.
  The reader may for instance try to prove that the inequality $P\neq Q$ of two second-order polynomials as functions implies that also $P^+\neq Q^+$.
  While a proof for the general case is not known to the authors, it is possible to prove this in the case where $P\neq Q$ is realized by a strictly monotone function argument.
  Note, that while it is not an unreasonable idea to restrict the domain of the second order polynomials, it should at least contain all (not necessarily strictly) monotone functions, as these show up as length functions of string functions.
  Just like for the general case, a proof of the above if the inequality is realized by a monotone function is not known to the authors.
  This leads to problems when trying to recursively define functions on the second-order polynomials.

  \subsection{Descriptions of second-order polynomials}

    This paper handles these difficulties by using descriptions of how to construct second-order polynomials instead.
    This section presents some results about second-order polynomials that are only needed to complete the proof of \Cref{resu:compositions} above and for the very end of the paper.
    On first reading it may be skipped and rolled back to when the results are needed.
   
    When constructing a second-order polynomial using the rules specified in the last section, it seems reasonable to bundle the uses of the \lq closure under addition\rq\ and the \lq closure under multiplication\rq\ rules that happen between two uses of the \lq application of the function argument\rq\ rule together to applying a multivariate polynomial.
    Formally this procedure can be described as follows:
    \begin{definition}
      A \demph{polynomial tree} is a finite tree $T$ whose nodes are elements of $\NN[X_0,\ldots,X_k]$ where $k$ coincides with the number of children the node has and there is a specified linear order on the children of each node.
    \end{definition}
    Given a polynomial tree, recursively assign to each node a second-order polynomial:
    To a leaf $t$ assign the second order polynomial $(l,n)\mapsto t(n)$.
    Now assume that second-order polynomials $P_1,\ldots,P_k$ were assigned to each of the children $t_1,\ldots,t_k$ of a node $t$.
    Assign to $t$ the second-order polynomial
    \[ (l,n)\mapsto t(n,l(P_1(l,n)),\ldots, l(P_k(l,n))) = t(n,P_1^+,\ldots, P_k^+). \]

    \begin{definition}
      A polynomial tree is called a \demph{description} of a second-order polynomial $P$ if $P$ is assigned to the root of the tree by the above procedure.
    \end{definition}
    \begin{minipage}{.65\textwidth}
      \vspace{.05cm}
      Note that there may exist many different descriptions of the same second-order polynomial.
      For instance both of the polynomial trees on the right hand side are bot descriptions of the  second-order polynomial $(l,n)\mapsto 2l(n)$.
      Whether or not these \quad\quad
      \vspace{-.26cm}
    \end{minipage}
    \begin{minipage}{.35\textwidth}
      \hfill
      \begin{tikzpicture}
        \node at (0,0) {$X_1+X_2$};
        \draw[->] (-.15,-.15) -- (-.35,-.85);
        \node at (-.5,-1) {$X_0$};
        \draw[->] (.15,-.15) -- (.35,-.85);    
        \node at (.5,-1) {$X_0$};
      \end{tikzpicture}
      \hfill
      \begin{tikzpicture}
        \node at (0,0) {$2 X_1$};
        \draw[->] (0,-.15) -- (0,-.85);
        \node at (0,-1) {$X_0$};
      \end{tikzpicture}
      \hfill
    \end{minipage}
    ambiguities can completely be avoided seems to be related to whether or not the operation $P\mapsto P^+$ is injective.

    An easy structural induction proves:
    \begin{lemma}
      Every second-order polynomial has a description.
    \end{lemma}
    \begin{proof}
      For the base case note that the a description consisting of a single node $p\in\NN[X_0]$ is a description of the second-order polynomial $(l,n)\mapsto p(n)$.

      To obtain a description of the point-wise sum $P+Q$ from descriptions of $P$ and of $Q$, let $t_P\in\NN[X_0,\ldots,X_k]$ be the polynomial at the root of $P$s description and $t_Q\in\NN[X_0,\ldots,X_{m}]$ the polynomial at the root of $Q$s description.
      A description of $P+Q$ is given by merging the root of the two descriptions to a node labeled with the polynomial
      \[ \tilde t(X_0,\ldots,X_{k+m+1}) := t_P(X_0,\ldots,X_k) + t_Q(X_{k+1},\ldots,X_{k+m+1}). \]

      For the point-wise product replace $t_P+t_Q$ in the above procedure by $t_P\cdot t_Q$.

      Finally note that if $P$ is a second order polynomial and $T$ a description of $P$ then adding a single node containing the polynomial $X_1$ above the root of $T$ is a description of $P^+$.
    \end{proof}

    This enables us to close a gap in the previous section:
    \begin{proposition}\label{resu:compositions}
      Whenever $P$ and $Q$ are second-order polynomials, then so are
      \[ (l,n)\mapsto P(Q(l,\cdot),n) \quad\text{and}\quad (l,n)\mapsto P(l,Q(l,n)). \]
    \end{proposition}
    \begin{proof}
      A description of the latter can be specified by replacing each leaf $p$ of a description of $P$ with a description of $Q$ where the root $t$ of $Q$ is replaced by $p\circ t$.
      For the former one each edge of in a description of $P$ has to be replaced with a description of $Q$ (where a copy of the part of the description of $P$ below the edge is appended to each leaf of the description of $Q$ and there are compositions again in the roots and the leafs).
    \end{proof}

    It would be desirable to find a distinguished description for each second-order polynomial.
    This would be a normal form theorem for second-order polynomials and in particular to make recursive definitions independent of the specific description and provide information about the second-order polynomial itself.
    The extend of ambiguity in descriptions is closely connected to injectivity of the mapping $P\mapsto P^+$.
    It seems to be impossible to use descriptions for the formulation of a normal for theorem unless injectivity holds.
    Some authors go as far as restricting to strictly monotone functions to force injectivity of functional application \cite{MR3219039}.

  \subsection{Polynomial majorants}\label{sec:polynomial majorants}
    As an example of a quantity that is well-defined on descriptions and of use later in the paper consider the following:
    \begin{definition}
      A pair $(N,p)$ of a natural number $N\in\NN$ and a function $p\colon\NN\to\NN$ is called a \demph{majorant} of a second order polynomial $P$ if $p(n)\geq n$ and there exists a description $T$ of $P$ such that
      \begin{itemize}
        \item $N$ is the height of the tree $T$.
        \item For each integer $n$ and each node $t$ of the tree $T$ it holds that $p(n)\geq t(n,\ldots,n)$.
      \end{itemize}
      A majorant is called a \demph{polynomial majorant} if $p$ is a polynomial.
    \end{definition}
    It is clear that each description of a second-order polynomial can be used to obtain a unique majorant by taking the minimal function that works for this description.
    A polynomial majorant can be constructed from a description choosing the coefficients of $p$ as maximum of the coefficients of the polynomials that arise from the nodes of the description by setting each of the variables to $n$.
    Since each second-order polynomial has a description.
    This proves:
    \begin{lemma}\label{resu:existence of a polynomial majorant}
      Any second-order polynomial has a polynomial majorant.
    \end{lemma}

    The following is the reason for the name \lq majorant\rq:
    \begin{lemma}\label{resu:property of a polynomial majorant}
      Let $(N,p)$ be a majorant of a second-order Polynomial $P$.
      Define a sequence of functions $p_i:\omega^\omega\times\omega\to\omega$ recursively by
      \[ p_0(l,n):= p(n)\quad\text{and}\quad p_{i+1}(l,n):=p(\max\{n,l(p_{i}(n))\}). \]
      Whenever $l:\omega\to\omega$ is monotone and $n\in\omega$ is arbitrary it holds that
      \[ P(l,n) \leq p_N(l,n). \]
    \end{lemma}
    \begin{proof}
      The proof proceeds by induction over the height of the description witnessing that $(N,p)$ is a polynomial majorant.

      For height $0$ the second order polynomial is of the form $(l,n)\mapsto q(n)$ for some polynomial $q$.
      By the assumption that $(N,p)$ is a polynomial majorant of $P$ it follows that
      \[ P(l,n) = q(n) \leq p(n) = p_0(l,n). \]

      Next assume that the statement has been proven for all descriptions of height $n<N$.
      Note that each of the $k$ children of the root can be regarded as a root of a description $T_k$ of a second-order polynomial $Q_k$.
      Each $T_k$ is a proper subtree of $T$, thus its height $n_k$ is strictly smaller than $Q_k$.
      From the induction hypothesis it follows that for all $l:\omega\to\omega$ and $n\in\omega$
      \[ Q_k(l,n)\leq p_{n_k}(l,n). \]
      Let $q$ be the polynomial at the root of $T$.
      Thus,
      \[ P(l,n) = q(n,l(Q_1(l,n)),\ldots,l(Q_k(l,n))) \]
      First note that for all monotone $l$ it holds that $p_i(l,n)\leq p_{i+1}(l,n)$.
      Since $(N,p)$ is a polynomial majorant of $P$ it holds that $p(m)\geq q(m,\ldots,m)$.
      Therefore, under the assumption that $l$ is monotone, it holds that
      \begin{align*}
        P(l,n) &\leq q(n,l(p_{N-1}(l,n)),\ldots,l(p_{N-1}(l,n))) \\ & \leq p(\max\{n,l(p_{N-1}(l,n))\}) \\ & = p_N(l,n).
      \end{align*}
      This proves the assertion.
    \end{proof}

  \subsection{Relativization}\label{sec:relativization}
    Second-order complexity theory usually only considers total functionals.
    However, the application we are most interested in is real complexity theory, which stems from computable analysis.
    In computable analysis, computations on continuous structures are carried out by encoding the objects by string functions.
    The mappings that assigns a \lq code\rq\ or \lq name\rq\ to the element it encodes are called representations.
    Computations on the space are then done by operating on the names instead.
    In this process, partial functionals are used.
    Recall the most basic notions from computable analysis.
    \begin{definition}
      A \demph{representation} $\xi$ of a space $X$ is a partial surjective mapping $\xi:\B\to X$.
    \end{definition}
    An element of $\xi^{-1}(x)$ is called a \demph{$\xi$-name} of $x$ or just a \demph{name}, if the representation is clear from the context.
    A pair $\XX=(X,\xi_{\XX})$ of a set and a representations of that set is called a represented space.

    Computations on represented spaces are carried out by operating on names:
    \begin{definition}
      Let $f:\XX\to \YY$ be a function between represented spaces.
      A partial functional $F:\subseteq\B\to\B$ is called a \demph{realizer} of $f$ if it translates $\xi_{\XX}$-names of $x$ to $\xi_{\YY}$-names of $f(x)$, that is if
      \[ \forall \varphi\in\dom(\xi_\XX): \xi_{\YY}(F(\varphi)) = f(\xi_{\XX}(\varphi)). \]
    \end{definition}
    A function is called \demph{computable} if it has a computable realizer.
    Here it is tradition not to make any assumptions about the behavior of the realizer outside of the domain of $\xi_{\XX}$.
    In particular the domain the domain of the realizer may be bigger than the domain of the representation and it may not have a total computable extension.

    Since we used the characterization by Kapron and Cook, it is possible to straightforwardly relax the definition of polynomial-time computability in an appropriate way.
    \begin{definition}
      Let $A\subseteq \B$.
      We say that an oracle Turing machine $M^?$ runs in \demph{$A$-restricted poly\-no\-mi\-al-time} if there exists a second-order polynomial $P$ such that for each oracle $\varphi$ from $A$ and string $\str a$ the computation of $M^\varphi(\str a)$ takes at most $P(\flength\varphi,\length{\str a})$ steps.
      I.e.
      \[ \forall \varphi\in A,\forall \str a\in\Sigma^*\colon \Time_{M^\varphi}(\str a)\leq P(\length{\varphi},\length{\str a}). \]
      We denote the set of functionals $F\colon A\to \B$ such that there is a machine computing $F$ in $A$-restricted polynomial time by $\operatorname P(A)$.
    \end{definition}
    Note that the requirement on $M^?$ in this definition has been weakened from having a polynomial running time (compare to \eqref{eq:RT}) by replacing the quantifier $\forall\varphi\in\B$ by a quantifier $\forall\varphi \in A$.
    
    Here are two examples of this definition covertly showing up in literature:
    \begin{example}[relativization]
    Oracle machines are used in classical complexity theory to talk about polynomial-time computability of a string function $\varphi\colon\Sigma^*\to\Sigma^*$ relative to some oracle $\psi:\Sigma^*\to \{\sdzero,\sdone\}$ interpreted as a subset of the strings.
    Under the assumption that $\psi$ only retruns $\sdzero$ or $\sdone$, one can check that the following are equivalent:
    \begin{itemize}
       \item  $\varphi$ is polynomial-time computable relative to $\psi$.
       \item The constant functional returning $\varphi$ is $\{\psi\}$-restricted polynomial-time computable.
     \end{itemize}
    This is the reason for the name of this chapter and remains true as long as $\psi$ has at most polynomial length.
    \end{example}
    
    The second example is Kawamura and Cook's framework for complexity for operators in analysis.
    Recall that Kawamura and Cook introduce the following subclass of Baire space:
    \begin{definition}[\cite{Kawamura:2012:CTO:2189778.2189780}]
      A string function $\varphi\in\B$ is called \demph{length-monotone} if for all strings $\str a$ and $\str b$ it holds that $\length{\str a} \leq \length{\str b}$ implies $\length{\varphi(\str a)}\leq \length{\varphi(\str b)}$.
      The set of all length-monotone string functions is denoted by $\reg$.
    \end{definition}
    Polynomial-time computability of functionals from $\reg$ to $\reg$ is then defined as $\reg$-restricted polynomial-time computability.
    (Of course it is not referred to by this name, but the definitions are identical.)
    Real complexity theory usually considers representations whose domains are included in the length-monotone string functions and regards a function between spaces that are equipped with such representations to be polynomial-time computable if it has a realizer that is polynomial-time computable in the above sense.

    The tradition in second-order complexity theory is to impose the running time requirement independently of the domain of the functional.
    \begin{definition}
      For $A\subseteq \B$ denote the class of all functionals $F\colon A\to \B$ that have a polynomial-time computable extension to all of Baire space by $\operatorname P|_A$.
    \end{definition}

    For a partial functional $F\colon A\to \B$ there are now two approaches to define polynomial-time computability.
    On one hand one could require that $F$ is $A$-restricted poly\-nomial-time computable, i.e., $F\in \operatorname P(A)$.
    On the other hand one could use the more restrictive definition that $F$ has a total polynomial-time computable extension, i.e., $F\in \operatorname P|_A$.
    The first definition follows the tradition of computable analysis, where no assumptions about a realizer are made outside of the domain of the representation on the input side of the operator.
    The second definition is in the tradition of second-order complexity theory, where one usually only considers polynomial-time computability of total functionals.
  \subsection{Incompatibility with relativization}
    Of course, the above distinction only makes sense if the classes $\operatorname P(A)$ and $\operatorname P|_A$ differ in general.
    Note that by definition $\operatorname P(A)\supseteq \operatorname P|_A$.
    Before we give the example that separates these classes, let us discuss why this result is not obvious.    
    The basic idea is to consider the length function on the string functions.
    Any oracle machine that computes this function takes a minimum of $2^n$ steps on any input of length $n$  and arbitrary oracle, as each query of length $n$ has to be asked to guarantee correctness of the return value.
    On the other hand, the brute-force search computes the length function in about $2^{n}l(n)$ time steps.
    This means, that the length function becomes $A$-restricted polynomial-time computable if $A$ is chosen as the set of string functions that have at least exponential length.
    
    Why does this not provide a counterexample already?
    Unfortunately, the brute force search can be modified to detect names of subexponential length and abort the computation in time.
    Informally such an algorithm can be described as follows:
    \lq Do a brute-force search, but abort as soon as you have to ask more than twice as many oracle queries as the length of the biggest return value you have found so far\rq.
    Such a machine does indeed compute the restriction of the length function on the exponentially growing functions while running in polynomial time for all inputs and returning something that differs from the length on the shorter functions (this is allowed since they are not in the domain).

    Thus, the argument has to be more elaborate.
    Our solution is to delay the time until a big input is provided:
    The elements of $A$ are only required to exhibit exponential growth on a sparse subset, i.e. $\length{\varphi}(g(n)) \geq 2^{g(n)}$, where $g$ is a fast growing function.
    Note that if $g$ does not grow fast enough, the trick above does still work.
    For instance for $g(n)=2^n$, the following algorithm still works:
    \lq Do a brute-force search but abort as soon as you have to ask more queries than the square of the biggest return value you have found so far\rq.
    If $g$ grows too fast the $A$-restricted polynomial-time computability may break down.

    Fortunately the choice $g(n)=2^{2^n}$ is a sweet spot:
    On one hand, due to the availability of length function iteration, it is still possible to use a second-order polynomial to extract a super exponential function from an element of the set therefore to make the brute-force algorithm work in $A$-restricted polynomial-time.
    On the other hand the above approach to compute a total extension does not work anymore and it becomes provable that no polynomial-time computable extension exists.

    \begin{theorem}[in general $\operatorname P|_A\subsetneq\operatorname{P}(A)$]\label{resu:failure of relativization for polytime}
      There exist a set $A\subseteq \B$ and a functional $F:A\to \B$ such that $F$ is $A$-restricted polynomial-time computable but has no total polynomial-time computable extension.
    \end{theorem}
    \begin{proof}
      Consider the set
      \[ A:= \{\varphi\in\B\mid \forall n\in\omega:\flength{\varphi}(2^{2^n})\geq 2^{2^{2^n}}\} \]
      and the functional on $A$ defined by
      \[ F:A\to\B,\quad F(\varphi)(\str a) := \sdzero^{\flength{\varphi}({\length{\str a}})}. \]
      $F$ is $A$-restricted polynomial-time computable.
      To see that this is true first note that $3(n+2)\geq 2^{2^{\lceil\lb(\lb(3(n+2))\rceil -1}}\in\NN$ and $\lb(n)^2\leq 3(n+2)$ (this is implied by the inequality $\ln(x)\leq \frac{x-1}{\sqrt x}$).
      Thus, for $\varphi\in A$ it holds that
      \begin{align*}
        \flength{\varphi}(\flength\varphi(3(n+2)))& \geq \flength{\varphi}(\flength{\varphi}(2^{2^{\lceil\lb(\lb(3(n+2)))\rceil-1}}))\geq \flength\varphi(2^{2^{2^{\lceil\lb(\lb(3(n+2)))\rceil-1}}}) \\
        & \geq 2^{2^{2^{2^{\lceil\lb(\lb(3(n+2)))\rceil-1}}}} \geq 2^{2^{\sqrt{3(n+2)}}}\geq 2^n
      \end{align*}
      This means that a second-order polynomial provides sufficient time to find the value of $\flength{\varphi}(\length{\str a})$ in $A$-restricted polynomial time using a brute-force search.

      However, $F$ does not have a total polynomial-time computable extension, as can be seen as follows:
      Towards a contradiction assume that there is an oracle Turing machine $M^?$ that computes such an extension in time bounded by some second-order polynomial $P$.
      For each $n\in\omega$ define an oracle $\varphi_n\in A$.
      First define a sequence of functions $\varphi_{n,k}\in\B$.
      Let $\varphi_{n,0}$ be the constant function returning $\varepsilon$.
      To recursively define $\varphi_{n,k+1}$ follow the computation $M^{\varphi_{n,k}}(\sdzero^{2^{2^n}-1})$ and whenever a query $\str a$ is asked such that $\lb(\lb(\length{\str a}))$ is an integer, then check whether all other queries of this length have been asked before and were answered with an $\varepsilon$ by $\varphi_{n,k}$.
      If this situation is encountered for some query $\str a_k$, then set $\varphi_{n,k+1}(\str a_k)$ to be the string of $2^{2^{2^m}}$ zeros, for all other strings $\str b$ set $\varphi_{n,k+1}(\str b):=\varphi_{n,k}(\str b)$ and ignore the rest of the computation.
      If such an $\str a_k$ does not exist, then set $\varphi_{n,k+1} := \varphi_{n,k}$.
      The sequence $(\varphi_{n,k})_k$ converges in Baire space, as the sequence $\varphi_{n,k}(\str a)$ is either constantly $\varepsilon$ or jumps to $\sdzero^{2^{\length{\str a}}}$ at some point and remains constant afterwards.
      Let $\tilde \varphi_n$ be the limit.
      Since $M^?$ is a deterministic machine, the computations $M^{\tilde \varphi_n}(\sdzero^{2^{2^n}-1})$ and $M^{\varphi_{n,k}}(\sdzero^{2^{2^n}-1})$ are identical up until the query $\str a_{k+1}$ is done.
      For the computation on oracle $\tilde \varphi_n$ to be finite, the sequence $(\str a_k)_k$ must be finite.
      Let $k_0$ be bigger than the number of elements, then $\tilde \varphi_n=\varphi_{n,k_0}$.
      Let $\varphi_n$ be the function that is identical to $\varphi_{n,k_0}$ unless $\varphi_{n,k_0}$ returns $\varepsilon$ on all inputs of length $2^{2^m}$.
      In this case not all the queries of this length were asked in the run of the machine $M^?$ on oracle $\varphi_{n,k_0}$ and input $\sdzero^{2^{2^n}-1}$.
      Pick one query of length $2^{2^m}$ that was not asked and let $\varphi_n$ return the string of $2^{2^{2^m}}$ zeros on this string.
      This guarantees that $\varphi_n\in A$.

      Let $\psi_n$ be the string function that coincides with $\varphi_n$ on strings of length less or equal $2^{2^{n-1}}$ (and thus also on all strings of length less or equal $2^{2^n}-1$) and returns $\varepsilon$ on bigger strings.
      Since the machine $M^?$ is deterministic for $M^{\varphi_n}(\sdzero^{2^{2^n}-1})$ and $M^{\psi_n}(\sdzero^{2^{2^n}-1})$ to differ it is necessary that an oracle query has been asked such that the answers of $\psi_n$ and $\varphi_n$ are distinct.
      The definition of $\varphi_n$ makes sure that this does not happen before all queries of length $2^{2^n}$ have been posed.
      Each of these queries takes one time step, thus $\Time_{M^{\psi_n}}(\str a) \geq 2^{2^{2^n}}$.
      If the machine runs identically on oracle $\varphi_n$ and oracle $\psi_n$, then it has to ask each query of length ${2^{2^n}-1}$ to correctly compute the length (otherwise we may change the value in the query of length $2^{2^n}-1$ that was not asked).
      Thus, for all $n$
      \[ \Time_{M^{\psi_n}}(\sdzero^{2^{2^n}-1}) \geq 2^{2^{2^n}-1}. \]
      
      By the definition of $\psi_n$ it holds that $\flength{\psi_n}(k)\leq 2^{2^{2^{n-1}}}$ for all $k$.
      Note that whenever $l$ is monotone and bounded by $m\in\omega$, i.e. $l(k)\leq m$ for all $k\in\omega$, then there exists a polynomial $p$ such that
      \[ P(l,k)\leq \max\{p(m),p(k)\}. \]
      Therefore,
      \[ P(\flength{\psi_{n}},2^{2^n}-1) \leq \max\{C2^{d2^{2^{n-1}}},p(2^{2^n}-1)\} \]
      holds for all $n$ and appropriate $C,d\in\omega$.
      Using that $P$ is a running time of $M^?$ and the inequality from above obtain
      \[ 2^{2^{2^n}-1} \leq \operatorname{time}_{M^{\varphi_n}}(\sdzero^{2^{2^n}-1}) = \operatorname{time}_{M^{\psi_n}}(\sdzero^{2^{2^n}-1}) \leq \max\{C2^{d2^{2^{n-1}}},p(2^{2^n}-1)\}. \]
      The maximum on the far right is assumed by the first term only for finitely many $n$: $2^{2^n}-1 \leq d 2^{2^{n-1}} + \lb(C)$ is a quadratic inequality for $x:=2^{2^{n-1}}$ and the set where it is fulfilled can be specified explicitly.
      However, this implies that the left hand side is bounded by a polynomial in $2^{2^n}-1$ which is clearly not the case.
      A contradiction.
    \end{proof}

    This proves that for an arbitrary set $A\subseteq \B$, it can not be expected that every $A$-restricted polynomial-time computable functional has a total polynomial-time computable extension.
    Does this mean that computable analysis uses a model that cannot be described by the usual approach of second-order complexity theory?
    Note that Kawamura and Cook replaced $\operatorname P|_{\reg}$ with $\operatorname P(\reg)$ in their framework.
    However, $\reg$ is far away from being an arbitrary set.

    Recall the following notion:
    \begin{definition}
      Let $A$ be a subset of $\B$.
      A mapping $R\colon \B\to A$ is called a \demph{retraction} of $\B$ onto $A$, if for all $\varphi\in A$ it holds that $R(\varphi) = \varphi$.
    \end{definition}
    A property of $\reg$ that guarantees the existence of total polynomial-time computable extensions is the following:
    \begin{lemma}
      There is a polynomial-time computable retraction from $\B$ onto $\reg$.
    \end{lemma}
    \begin{proof}
      For a string $\str a$ let $\str a^{\leq n}$ denote its initial segment of length $n$ (or the string itself if it has less than $n$ bits).
      Consider the mapping
      \[ R(\varphi)(\str a) := \varphi(\str a)^{\leq \length{\varphi(\sdzero^n)}}\sdzero^{\max\{\length{\varphi(\sdzero^n)}-\length{\varphi(\str a)},0\}}. \]
      This mapping is a polynomial-time computable retraction from $\B$ onto $\Sigma^{**}$.
    \end{proof}

    \begin{theorem}
      Whenever there is a polynomial-time computable retraction from $\B$ onto $A$, then any $A$-restricted polynomial-time computable functional has a total polynomial-time computable extension.
      I.e.\ $\operatorname{P}|_{A}=\operatorname{P}(a)$.
    \end{theorem}

    \begin{proof}
      The proof that the composition of two polynomial-time computable functionals is polynomial-time computable from \Cref{resu:composition} remains valid if the assumptions are weakened to $F$ being $G(\B)$-restricted polynomial-time computable.
      Thus, the composition of the $A$-restricted polynomial-time computable functional with the retraction is polynomial-time computable.
    \end{proof}
    The previous two results directly entail the following:
    \begin{corollary}[$\operatorname P(\reg) = \operatorname P|_{\reg}$]
      A functional $F:\reg\to\B$ is polynomial-time computable in the sense of Kawamura and Cook if and only if it has a total polynomial-time computable extension.
    \end{corollary}
    An alternative proof can be obtained by adding a clock to the machine.
    (Details about how to clock such a machine can be found in the proof of Theorem \ref{resu:spreg equals preg}.)

\section{Query dependent step restrictions}

  In this section, we investigate a different approach to measuring the running time of an oracle machine that does not rely on higher order objects as running times.
  Recall that for a regular Turing machine the time function $\timef_M:\Sigma^*\to \omega$ is defined to return on input $\str a$ the number of steps that it takes until the machine terminates on input $\str a$.
  A running time of the machine is then defined to be a function $t:\omega\to\omega$ such that
  \[ \tag{rt}\label{eq:rt} \forall \str a\in\Sigma^*: \timef_M(\str a)\leq t(\length{\str a}). \]
  For an oracle Turing machine, each of the time functions $\timef_{M^\varphi}$ may be different.
  Thus, the above definition has to be replaced.
  The most common replacement is to replace $t$ by a higher type object as discussed in the previous section.
  However, there exist other approaches of how to replace this definition in literature.
  Some of them stay with functions of type $\omega\to\omega$ for running times.
  So does the notion this part of the paper introduces.
  To distinguish these objects from the time function and the running times from second-order complexity theory, we refer to such objects as \lq step-counts\rq\ instead of \lq running times\rq.

  One example of a definition in this vein has been investigated by Stephen Cook \cite{MR1159205}.
  He bounds the steps an oracle Turing machine may take by modifying \eqref{eq:rt} as follows:
  He replaces $\length{\str a}$ by the maximum $m_{M^\varphi,\str a}$ of $\length{\str a}$ and the biggest length of any of the oracle answers in the run of $M^?$ with oracle $\varphi$ on input $\str a$ and additionally universally quantifies over $\varphi\in\B$.
  Thus, ending up with
  \[ \forall\varphi\in\B,\forall\str a\in\Sigma^*\colon \Time_{M^\varphi}(\str a) \leq t(m_{M^\varphi}). \]
  He refers to the class of functionals that can be computed by a machine fulfilling the above for $t$ being some polynomial as $\opt$ (for \lq oracle polynomial time\rq).

  \subsection{Step-counts}

    We use a slightly more complicated definition that turns out to be considerably more well-behaved.
    \begin{definition}
      Let $M^?$ be an oracle Turing machine.
      For a given oracle $\varphi$ and a given input $\str a$ denote the content of the oracle answer tape in the $k$-th step of the computation by $\str b_k$.
      Define the \demph{length revision function} $o_{\varphi,\str a}:\omega\to\omega$ recursively as follows:
      \[ o_{\varphi,\str a}(0) := \length{\str a} \quad\text{and}\quad o_{\varphi,\str a}(n+1) :=\max\{o_{\varphi,\str a}(n), \length{\str b_{n+1}}\}. \]
    \end{definition}
    Note that $o_{\varphi,\str a}(k+1)>o_{\varphi,\str a}(k)$ means that in the $k$-th step of the computation, the machine asks an oracle query and the answer is bigger than both the input $\str a$ and any of the answers the oracle has given earlier in the computation.
    We call this a \demph{length revision} as it means that it became apparent to the machine that its input (the oracle) is bigger than what the previous evidence indicated.

    For an oracle machine $M^?$ with a fixed oracle $\varphi\in\B$ let $\timef_{M^\varphi}(\str a)\in\omega \cup\{\infty\}$ be the number of steps that the computation of $M^\varphi$ takes on input $\str a$.
    I.e. the machine is explicitly allowed to diverge on some inputs.
    \begin{definition}[compare \cref{fig:step-count}]\label{def:step-count}
      A function $t\colon\omega\to\omega$ is a \demph{step-count} for an oracle Turing machine $M^?$ if
      \[ \forall \varphi\in\B,\forall \str a\in\Sigma^*, \forall n\leq \timef_{M^\varphi}(\str a)\colon n\leq t(o_{\varphi,\str a}(n)). \]
      Denote the set of all functionals on the Baire space that can be computed by an oracle Turing machine that has a polynomial step-count by $\PSC$.
    \end{definition}
    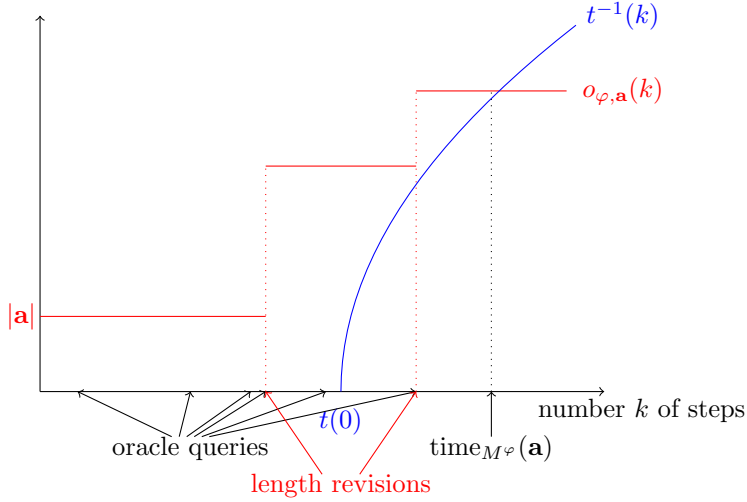
\begin{figure}
      \centering
      \begin{tikzpicture}
        \draw[->] (0,0) -- (0,5);
        \draw[->] (0,0) -- (7.5,0);
        \node at (8,-.25) {number $k$ of steps};
        \draw[color=red] (0,1) -- (3,1);
        \draw[color=red] (3,3) -- (5,3);
        \draw[color=red] (5,4) -- (7,4);
        \node at (6,-.75) {$\timef_{M^\varphi}(\str a)$};
        \draw[->] (6,-.6) -- (6,0);
        \draw[domain = 0:1.25,blue] plot ({2*\x*\x+4},{3.9*\x});
        \node[blue] at (4,-.4) {$t(0)$};
        \node[blue] at (7.75,5) {$t^{-1}(k)$};
        \node[red] at (7.75,4) {$o_{\varphi,\str a}(k)$};
        \draw[dotted] (6,0) -- (6,4);
        \node[red] at (-.25,1) {$\length{\str a}$};
        \node at (2,-.75) {oracle queries};
        \draw[->] (1.75,-.6) -- (.5,0);
        \draw[->] (1.85,-.6) -- (2,0);
        \draw[->] (1.95,-.6) -- (2.8,0);
        \draw[->] (2.05,-.6) -- (3,0);
        \draw[->] (2.15,-.6) -- (3.8,0);
        \draw[->] (2.25,-.6) -- (5,0);
        \node[red] at (4,-1.25) {length revisions};
        \draw[->,red] (3.75,-1.1) -- (3,0);
        \draw[->,red] (4.25,-1.1) -- (5,0);
        \draw[dotted,red] (3,0) -- (3,3);
        \draw[dotted,red] (5,0) -- (5,4);
      \end{tikzpicture}
      \caption{Verifying that $\varphi$ and $\str a$ are not a counterexample of $t$ being a step-count.
      Under the assumption that $t$ is invertible on the set $[t(0),\infty)$.}\label{fig:step-count}
    \end{figure}
    Note that in contrast to \Cref{eq:rt}, the above is not void if the machine diverges on some inputs.
    The relationship between termination of a machine and the existence of a step-count is quite involved.
    For instance: If a machine has a step-count and diverges, then the machine queries the oracle an infinite number of times.
    Furthermore, if there is an integer bound on the length of all return values of an oracle, then every machine that has a step-count terminates when given that oracle and an arbitrary input.

    Note that $o_{\varphi,\str a}(\Time_{M^\varphi}(\str a))$ is by definition the maximum of the length of $\str a$ and the biggest oracle query done in the computation of $M^\varphi(\str a)$.
    This number was previously called $m_{M^\varphi,\str a}$.
    Thus, Stephen Cook's class $\opt$ can be reproduced by not quantifying over all $n\leq\Time_{M^\varphi}(\str a)$ but only considering the case $n=\timef_{M^\varphi}(\str a)$.
    In upcoming proofs it is used that it is possible to clock a machine while basically maintaining the same step-count by checking in each step, that the requirement above is fulfilled.
    Note that this is not possible for the machines used by Cook without increasing the step-count considerably, as his framework allows to retroactively justify high time-consumption early in the computation by a big oracle answer late in the computation.

    The very example that Cook used to disregard the class $\opt$ as a candidate for the class of polynomial-time functionals can be used to also disregard the class of total functionals that are computed by a machine that allows a polynomial step-count:
    \begin{example}[$\PSC\not\subseteq \operatorname{P}$]
      The total functional $F:\B\to\B$ defined by
      \[ F(\varphi)(\str a) := \varphi^{\length{\str a}}(\sdzero) \]
      can be computed by an oracle Turing machine that has a polynomial step-count but does not carry polynomial-time computable input to polynomial-time computable output.

      To see that this machine has a polynomial step-count, note that it can be computed by the machine that proceeds as follows:
      It copies the input to the memory tape and writes $\sdzero$ to the oracle query tape.
      Then as long as the memory tape is not empty it repeats the following steps:
      First copies the content of the oracle answer band to the oracle query band.
      Then it removes the content of the last non-empty cell from the memory band.
      Finally it enters the oracle query state.
      When the memory tape is empty it copies the content of the oracle answer band to the output tape and enters the termination state.

      Copying a string of length $n$ takes $\bigo(n)$ steps.
      The length of the string that has to be copied is always bounded by the previous oracle answers.
      The loop is carried out exactly $\length{\str a}$ times.
      Therefore, there is some step-count in $\bigo(n^2)$.

      To verify that the functional does not preserve the class of polynomial-time computable functionals consider the polynomial-time computable functional $\psi(\str a) := \str a\str a$.
      Note that
      \[ F(\psi)(\str a) = \psi^{\length{\str a}}(\sdzero) = \sdzero^{2^{\length{\str a}}}. \]
      Therefore, writing $F(\psi)(\str a)$ takes at least $2^{\length{\str a}}$ steps and thus $F(\psi)$ cannot be polynomial-time computable.
    \end{example}
    This means that further restrictions are necessary.
    In \cite{MR1159205} this is the point where Stephen Cook decides to use polynomial-time computable functionals.
    This paper presents a different set of restrictions that can be used.

  \subsection{Finite length-revision}

    Let $M^?$ be an oracle Turing machine that always terminates.
    Then for any oracle $\varphi$ and any string $\str a$ the computation of $M^\varphi$ on $\str a$ is finite and only queries the oracle a finite number of times.
    Note that $\#o_{\varphi,\str a}(\omega)$ coincides with the number of length revisions that happens during the computation on oracle $\varphi$ and input $\str a$.
    Since the number of length revisions is bounded by the number of total oracle queries, the following statement holds true:
    \[ \forall \varphi\in \B,\forall \str a\in\Sigma^*\colon \exists N\in\NN\colon \#o_{\varphi, \str a}(\omega)\leq N. \]
    In general $N$ depends on the choice of the oracle and the string.
    Our restriction on the behavior of the machine is that there is an $N$ that works independently of the choice of the oracle and the input.
    \begin{definition}\label{def:finite length-revision}
      We say that an oracle Turing machine $M^?$ has \demph{finite length-revision} if there is an integer $N$ such that no matter what the oracle and the input are, no more than $N$ length revisions happen.
      That is, if its length revision functions $o_{\varphi,\str a}$ fulfill
      \[ \exists N\in\NN\colon \forall \varphi\in \B, \forall \str a\in\Sigma^*\colon\#o_{\varphi,\str a}(\omega) \leq N. \]
      We denote the set of all functionals on the Baire space that can be computed by machines with finite length-revision by $\FLR$.
    \end{definition}
    Finite length revision does a priori neither restrict the number of oracle questions nor the length of the oracle answers:
    The restriction is that there is a finite number of length revisions, that is, only a finite number of times it happens that a query is asked such that the answer is strictly bigger than the input and any earlier oracle answer.

    \begin{example}[$\operatorname{P}\not\subseteq \FLR$]\label{ex:flr not contained in p}
      Consider the functional
      \[ F:\B\to\B, \quad F(\varphi)(\str a):= \sdzero^{\max\{\length{\varphi(\sdzero^n)}\mid n\leq \length{\str a} \}}. \]
      The straightforward implementation asks $n$ queries, compares their lengths and returns the maximum.
      This can be done in time $P(l,n) = C(n + n \cdot l(n))+C$ for some $C\in\NN$.
      However, since $\length{\varphi(\sdzero^n)}$ may be strictly increasing when $n$ increases, this machine does not have finite length revision.
      
      Indeed, no machine with finite length revision can compute $F$, as can be proven via contradiction as follows:
      Assume that there was such a machine $M^?$.
      Let $N$ be a bound on the length-revisions $M^?$ does.
      Define an oracle such that the output of $M^\varphi(\sdzero^{N+1})$ is incorrect as follows:
      Let $\str a_1$ be the first oracle query that is asked in the run of the machine $M^\varphi(\sdzero^{N})$.
      Set $\varphi(\str a_1) := \sdzero^{N+1}$.
      Thus, a length-revision happens.
      Let $\str a_2$ be the next oracle query that the machine poses.
      Set $\varphi(\str a_2):= \sdzero^{N+2}$.
      This means that another length revision happens.
      Carry on in that way until $\varphi(a_N)$ is set to $\sdzero^{2N}$.
      After asking the query $\str a_N$, the machine can not ask another query as we may as well set the return value to be bigger again and no further length revision is allowed.

      Note that the run of the machine on $\sdzero^N$ is identical for any oracle that fulfills $\psi(\str a_i)=\sdzero^{N+i}$.
      Let $M$ be the number of steps the machine $M^\psi$ takes for any of these oracles to terminate.
      There are $N+1$ strings of the form $\sdzero^n$ for $n\leq N$.
      Thus, at least one of these strings is not contained within $\str a_1,\ldots,\str a_N$.
      Let $\sdzero^m$ be this string.
      Let $\varphi$ be the string function defined as follows:
      \[ \varphi(\str b) = \begin{cases} \sdzero^{N+i} &\text{if }\str b=\str a_i\\ \sdzero^{M+1} &\text{if }\str b = \sdzero^m\\\varepsilon &\text{otherwise.} \end{cases} \]
      Obviously, the run of $M^\varphi$ on $\sdzero^N$ coincides with the one described above.
      Therefore the return value can have at most $M$ bits.
      Since $m\leq N$ it holds that $\length{F(\varphi)(\sdzero^N)} \geq \length{\varphi(\sdzero^m)} \geq M+1$.
      Thus $M^\varphi$ can on input $\sdzero^N$ not produce the right return value.
    \end{example}

  \subsection{Strong polynomial-time computability}

    While neither finite length revision nor having a step-count implies termination of the machine, the combination does:
    We mentioned that a machine that has a step-count may only diverge with oracle $\varphi$ if there is no bound on the oracle answers.
    This, however, is forbidden by finite length revision.
    Therefore, if $M^?$ is a machine that has finite length-revision and a step-count, then the computation of $M^?$ with any oracle and on any input terminates.

    \begin{definition}
      Call a functional $F\colon \B \to\B$ \demph{strongly polynomial-time computable} if there is an oracle Turing machine computing $F$ that has both finite length-revision and a polynomial step-count (see Definition~\ref{def:step-count}).
      We denote the set of all strongly polynomial-time computable operators by $\operatorname{SP}$
    \end{definition}

    As the name suggests, strong polynomial-time computability implies poly\-no\-mial-time computability.
    \begin{lemma}[$\operatorname{SP}\subseteq \operatorname{P}$]\label{resu:strongpolytime implies polytime}
      Any total strongly poly\-nom\-ial-time computable functional is polynomial-time computable.
    \end{lemma}

    \begin{proof}
      Let $M^?$ be the Turing machine that verifies that the total functional is strongly polynomial time computable, $p$ a polynomial step-count of the machine and $N$ a bound of the number of length revisions it does.
      To see that the machine runs in polynomial time fix some arbitrary oracle $\varphi$ and a string $\str a$.
      By the definition of being a step-count, the first oracle query in the run of $M^\varphi$ on input $\str a$ has at most $p(\length{\str a})$ bits.
      Thus the return value of the oracle has at most length $\flength{\varphi}(p(\length{\str a}))$.
      Therefore, again since $p$ is a step-count, the next oracle query that leads to a length revision can not have more than $p(\flength{\varphi}(p(\length{\str a}))+\length{\str a})$ bits.
      Repeating the above argument $N$ times and using that $N$ is a bound of the number of length revisions proves that the computation terminates within at most $(p\circ(\flength{\varphi}+\id))^N(p(\length{\str a}))$ steps.
      That is, that the second order polynomial $P(l,n):= (p\circ (l+\id))^N(p(n))$ is a running time of $M^?$.
    \end{proof}
    Note that a better time bound of the machine is given by the function $p_N$ defined just as in \Cref{resu:property of a polynomial majorant}.
    However, this function is in general not a second order polynomial.

    On the other hand, strong polynomial-time computability is a strictly stronger requirement than polynomial-time computability.
    \begin{lemma}[$\operatorname{SP}\subsetneq \operatorname{P}$]
      There exists a polynomial-time computable functional that is not computable with finite length-revision.
      In particular, this functional is not strongly polynomial-time computable.
    \end{lemma}

    \begin{proof}
      An functional that is polynomial-time computable but not computable with finite length revision was discussed in detail in Example~\ref{ex:flr not contained in p}.
      Since $\operatorname{SP}\subseteq \FLR$, this indeed proves that the inclusion $\operatorname{SP}\subseteq \operatorname{P}$ from the previous result is strict.
    \end{proof}

    A candidate for a natural example of an operator from analysis that is not strongly polynomial-time computable is constructed in \cite{Aminimal}.

  \subsection{Compatibility with relativization}\label{sec:compatibility with relativization}

    For strong polynomial-time computability, relativized notions can be introduced analogously to Section~\ref{sec:relativization}:
    Let $A\subseteq \B$.
    A machine $M^?$ is said to run in $A$-restricted strongly polynomial time if the number of length revisions $M^?$ does on oracles from $A$ is bounded by a number and there is a polynomial step-count that is valid whenever the oracle is from $A$.
    That is if the formulas from Definition~\ref{def:step-count} and Definition~\ref{def:finite length-revision} are fulfilled if \lq$\forall \varphi\in\B$\rq\ is replaced by \lq$\forall\varphi\in A$\rq.
    Again, we denote the set of all functionals whose domain is $A$ and that can be computed by an $A$-restricted strong polynomial-time machine by $\operatorname{SP}(A)$ and the set of all functionals whose domain is contained in $A$ and that have a total strongly polynomial-time computable extension by $\operatorname{SP}|_A$.
    For strong polynomial-time computability these classes coincide.
    This may be interpreted as strong polynomial-time computability being more well behaved with respect to partial functionals.

    \begin{lemma}[$\operatorname{SP}(A)=\operatorname{SP}|_A$]\label{resu: sp(a) equals spla}
      A $A$-restricted strongly poly\-nomial-time computable functional has a total strongly polynomial-time computable extension.
    \end{lemma}

    \begin{proof}
      Let $F:A\to\B$ be an $A$-restricted strongly polynomial-time computable functional and let $M^?$ be a machine that witnesses the strong polynomial-time computability of the functional.
      Let $N$ be maximum number of length revisions $M^?$ does on any oracle from $A$ and let $p$ be a polynomial step-count valid for input from $A$.
      Define a new machine $\tilde M^?$ as follows:
      $\tilde M^?$ starts by initializing a counter with $N$ written on it.
      Furthermore it saves the length of the input string and produces the coefficients of $p$ on the memory tape.
      It applies the polynomial $p$ to the length of the input and initializes a second counter holding this value.
      Now it follows the exact same steps $M^?$ does as long as no oracle query is done and meanwhile counts down the second counter.
      If the second counter hits zero, it terminates and returns $\varepsilon$.
      If before that happens, an oracle call is done, it decreases the first counter.
      If the counter was already zero, it terminates and returns $\varepsilon$.
      If it was not, it writes the maximum of the previous content and the length of the return value to where it originally noted the length of the input.
      It applies the polynomial to this new value and ads the difference to the previous value to the second counter.
      Then it continues as before.

      It is clear that the machine described above runs with length revision $N+1$, that it has a polynomial step-count (that depends only on $p$ and $N$) and that whenever the oracle is from $A$, none of the counters will hit zero and $\tilde M^\varphi$ and $M^\varphi$ produce the same values in the end.
      Thus $\tilde M^?$ computes a total strongly polynomial-time computable extension of $F$.
    \end{proof}
    This proves that there is a stable notion of strong polynomial-time computability of partial functionals.
    In particular referring to partial functionals as being strongly polynomial-time computable does not lead to confusion and we may drop the \lq$A$-restricted\rq\ part.

  \subsection{Comparison to polynomial-time on $\reg$}

    Recall that originally polynomial-time computability was only defined for machines that compute total functions.

    Kawamura and Cook's framework for complexity of operators in analysis, however, does not require a realizer to have a total polynomial-time computable extension, but instead gives a new definition of what polynomial-time computability of a functional on $\reg$ means.
    Earlier, this notion of complexity was called being $\reg$-restricted polynomial-time computable and the class of these functionals was denoted by $P(\reg)$.
    This section proves that a functional $F:A\to\B$ whose domain is contained in $\reg$, is $A$-restricted polynomial-time computable if and only if it is strongly polynomial time computable.
    Note that here strong polynomial time computability means one of the two conditions of being from $\operatorname{SP}(A)$ or from $\operatorname{SP}|_A$  that were proven equivalent in \Cref{resu: sp(a) equals spla}.
    In particular this equates all the intermediate classes, like those functionals that have an extension from $\operatorname{P}(\reg)$.
    Furthermore, it implies that the domain of the functional considered in Example~\ref{ex:flr not contained in p} was necessarily not contained in $\reg$.

    \begin{theorem}[$A\subseteq\reg\Rightarrow\operatorname{SP}(A) = \operatorname{P}(A)$]\label{resu:spreg equals preg}
      Let $A\subseteq \reg$.
      A functional is $A$-restricted poly\-no\-mial-time computable if and only if it is strongly polynomial-time computable.
    \end{theorem}

    \begin{proof}[That $\operatorname{SP}(A)\subseteq \operatorname{P(A)}$]
      This direction follows from previous results:
      Let $F\colon A\to\B$ be an $A$-restricted polynomial-time computable operator.
      \Cref{resu: sp(a) equals spla} implies that $F$ has a total strongly polynomial-time computable extension.
      By \Cref{resu:strongpolytime implies polytime}, this total extension is polynomial-time computable.
      In particular it is $A$-restricted polynomial-time computable, as this is a weaker requirement.
      Therefore it is contained in $\operatorname P(A)$.
    \end{proof}

    The other direction of the proof heavily relies on the notions discussed in \Cref{sec:polynomial majorants}.

    \begin{proof}[That $\operatorname{P}(A)\subseteq\operatorname{SP}(A)$]
      Let $F:A\to\B$ be computable in $A$-restricted polynomial time.
      By \Cref{resu: sp(a) equals spla} this operator has a total polynomial time computable extension.
      Let $M^?$ be a machine that computes this extension in time bounded by a second-order polynomial $P$.
      From \Cref{resu:existence of a polynomial majorant} it follows that there exists a polynomial majorant $(N,p)$ of $P$.
      Define a new oracle Machine $\tilde M^?$ as follows:
      When given $\varphi$ as oracle and a string $\str a$ as input, the machine computes $p(m)$ with $m:=\length{\str a}$.
      It then poses the oracle query $\varphi(\sdzero^{p(m)})$ and takes the maximum of the length of the return value and $m$.
      It repeats this procedure with $m$ set to be this maximum.
      The above is repeated $N$ times.
      It writes the result into a counter, does a final query of the oracle on the value of this counter many zeros and then caries out the computations $M^?$ does on oracle $\varphi$ and input $\str a$ while counting the counter down.
      If the counter runs empty or a length revision is encountered it terminates and returns $\varepsilon$.
      If $M^?$ terminates without this happening, it returns $M^\varphi(\str a)$.

      Whenever the oracle $\varphi$ is length monotone, the above procedure is easily checked to first produce a value of the function $p_N(\length{\varphi},\length{\str a})$ from \Cref{resu:property of a polynomial majorant} thereby doing at most $N-1$ length revisions and within a polynomial step-count.
      Then it simulates the machine $M^?$, for at most $p_N(\length{\varphi},\length{\str a})$ steps and allowing at most one furter length revision.
      Thus, the machine $\tilde M^?$ runs in strongly polynomial time.
      Since $p_N(\length{\varphi},\length{\str a})\geq P(\length{\varphi},\length{\str a})$ by \Cref{resu:property of a polynomial majorant}, the simulation comes to an end before the timer is empty whenever $\varphi$ is from the domain of $F$.
      This proves that $\tilde M^\varphi$ computes a total strongly polynomial-time computable extension of $F$.
      In particular $F\in\operatorname{SP}(A)$.
    \end{proof}

\section{Conclusion}
  
  The results of this paper are tightly connected to questions of whether or not it is possible to add clocks to certain machines.
  Clocking is a standard procedure to increase the domain of machines while maintaining its behavior on a set of \lq important\rq\ oracles and inputs.
  For regular Turing machines, clocking allows to turn any machine that runs in polynomial time on the inputs the user cares about into a machine that actually runs in polynomial time:
  Take the polynomial that bounds the running time on the important inputs and in each step check if this number of steps was exceeded.
  This machine runs in about the same time as the original machine due to the time constructibility of polynomials.
  When moving to oracle Turing machines, the polynomials have to be replaced by second-order polynomials and unfortunately, these turn out not to be time constructible.
  Thus, for oracle Turing machines the above procedure does not extend in a straight forward manner.
  Indeed, Theorem~\ref{resu:failure of relativization for polytime} proves that it is in principle impossible to clock a polynomial-time machine in general.
  This can be understood as a very strong version of the failure of time-constructibility of second-order polynomials.

  One of the main motivations Kawamura and Cook had when they restricted the domains of the functionals they considered to be length-monotone functions was to force clockability of polynomial-time machines  \cite{Kawamura:2012:CTO:2189778.2189780}.
  And indeed, in this framework the second-order polynomials can be proven time-constructible \cite{ArXiV}.
  The notion of strong polynomial-time computability introduced in this paper tackles the same problem from another angle:
  It introduces a subclass of the polynomial-time functionals such that clocking is possible on any domain.
  However, strong polynomial-time computability is a strictly stronger condition than polynomial-time computability.
  
  We hope that strong polynomial-time computability turns out to be a useful concept.
  We think it has potential for usefulness and that it is a further step towards the expectations of programmers what programs with subroutine calls should be considered fast as it removes the dependency of the running time on information that can not be read from the oracle in a fast way.

  \bibliography{bib}{}

\begin{thebibliography}{KSZ16b}

\bibitem[BS17]{Aminimal}
Franz Brau{\ss}e and Florian Steinberg.
\newblock A minimal representation for continuous functions.
\newblock \url{https://arxiv.org/abs/1703.10044}, 2017.
\newblock preprint.

\bibitem[Coo91]{MR1159205}
Stephen~A. Cook.
\newblock Computational complexity of higher type functions.
\newblock In {\em Proceedings of the {I}nternational {C}ongress of
  {M}athematicians, {V}ol.\ {I}, {II} ({K}yoto, 1990)}, pages 55--69. Math.
  Soc. Japan, Tokyo, 1991.

\bibitem[FGH14]{MR3239272}
Hugo F{\'e}r{\'e}e, Walid Gomaa, and Mathieu Hoyrup.
\newblock Analytical properties of resource-bounded real functionals.
\newblock {\em J. Complexity}, 30(5):647--671, 2014.
\newblock \href {http://dx.doi.org/10.1016/j.jco.2014.02.008}
  {\path{doi:10.1016/j.jco.2014.02.008}}.

\bibitem[FH13]{higherorder}
Hugo Férée and Mathieu Hoyrup.
\newblock Higher order complexity in analysis, 2013.
\newblock CCA.
\newblock URL: \url{https://hal.inria.fr/hal-00915973/document}.

\bibitem[FZ15]{postive}
Hugo Férée and Martin Ziegler.
\newblock On the computational complexity of positive linear functionals on
  c[0;1], 2015.
\newblock MACIS conference.
\newblock URL: \url{https://hugo.feree.fr/macis2015.pdf}.

\bibitem[Kaw11]{kawamuraphd}
Akitoshi Kawamura.
\newblock {\em Computational Complexity in Analysis and Geometry}.
\newblock PhD thesis, University of Toronto, 2011.

\bibitem[KC96]{MR1374053}
B.~M. Kapron and S.~A. Cook.
\newblock A new characterization of type-{$2$} feasibility.
\newblock {\em SIAM J. Comput.}, 25(1):117--132, 1996.
\newblock \href {http://dx.doi.org/10.1137/S0097539794263452}
  {\path{doi:10.1137/S0097539794263452}}.

\bibitem[KC12]{Kawamura:2012:CTO:2189778.2189780}
Akitoshi Kawamura and Stephen Cook.
\newblock Complexity theory for operators in analysis.
\newblock {\em ACM Trans. Comput. Theory}, 4(2):5:1--5:24, May 2012.
\newblock \href {http://dx.doi.org/10.1145/2189778.2189780}
  {\path{doi:10.1145/2189778.2189780}}.

\bibitem[KP14]{MR3219039}
Akitoshi Kawamura and Arno Pauly.
\newblock Function spaces for second-order polynomial time.
\newblock In {\em Language, life, limits}, volume 8493 of {\em Lecture Notes in
  Comput. Sci.}, pages 245--254. Springer, Cham, 2014.
\newblock \href {http://dx.doi.org/10.1007/978-3-319-08019-2_25}
  {\path{doi:10.1007/978-3-319-08019-2_25}}.

\bibitem[KSZ16a]{Kawamura:2016:CTC:2933575.2935311}
Akitoshi Kawamura, Florian Steinberg, and Martin Ziegler.
\newblock Complexity theory of (functions on) compact metric spaces.
\newblock In {\em Proceedings of the 31st Annual ACM/IEEE Symposium on Logic in
  Computer Science}, LICS '16, pages 837--846, New York, NY, USA, 2016. ACM.
\newblock \href {http://dx.doi.org/10.1145/2933575.2935311}
  {\path{doi:10.1145/2933575.2935311}}.

\bibitem[KSZ16b]{CIE2016}
Akitoshi Kawamura, Florian Steinberg, and Martin Ziegler.
\newblock Towards computational complexity theory on advanced function spaces
  in analysis.
\newblock In Arnold Beckmann, Laurent Bienvenu, and Nata{\v{s}}a Jonoska,
  editors, {\em Pursuit of the Universal: 12th Conference on Computability in
  Europe, CiE 2016, Paris, France, June 27 - July 1, 2016, Proceedings}, pages
  142--152. Springer International Publishing, Cham, 2016.
\newblock \href {http://dx.doi.org/10.1007/978-3-319-40189-8_15}
  {\path{doi:10.1007/978-3-319-40189-8_15}}.

\bibitem[Lam06]{MR2275414}
Branimir Lambov.
\newblock The basic feasible functionals in computable analysis.
\newblock {\em J. Complexity}, 22(6):909--917, 2006.
\newblock \href {http://dx.doi.org/10.1016/j.jco.2006.06.005}
  {\path{doi:10.1016/j.jco.2006.06.005}}.

\bibitem[Meh76]{MR0411947}
Kurt Mehlhorn.
\newblock Polynomial and abstract subrecursive classes.
\newblock {\em J. Comput. System Sci.}, 12(2):147--178, 1976.
\newblock Sixth Annual ACM Symposium on the Theory of Computing (Seattle,
  Wash., 1974).

\bibitem[SS17]{ArXiV}
Matthias Schr{\"o}der and Florian Steinberg.
\newblock Bounded time computation on metric spaces and {B}anach spaces.
\newblock \url{https://arxiv.org/abs/1701.02274}, 2017.
\newblock preprint; extended abstract accepted for LICS 2017 conference.

\end{thebibliography}
\end{document}